\documentclass{article}

\usepackage{microtype}
\usepackage{amsthm}
\usepackage{amsmath, amsfonts}
\usepackage{graphicx}
\usepackage{todonotes}
\usepackage{theoremref}
\usepackage{hyperref}
\usepackage{thmtools}
\usepackage{thm-restate}

\usepackage{overpic}

\usepackage{caption}

\newcommand{\ket}[1]{\mathop{\left|#1\right>}\nolimits}
\newcommand{\bra}[1]{\mathop{\left<#1\,\right|}\nolimits}
\newcommand{\ketn}[1]{\mathop{|#1\rangle}\nolimits} 

\newcommand{\akb}[1]{| #1\rangle\!\langle #1 |}

\newcommand{\COMMENT}[1]{}

\DeclareMathOperator{\Tr}{Tr}

\DeclareMathOperator{\PP}{\mathbb{P}}

\DeclareMathOperator{\QIC}{\mathrm{QIC}}

\DeclareMathOperator{\QIL}{\mathrm{QIL}}

\DeclareMathOperator{\QCC}{\mathrm{QCC}}

\DeclareMathOperator{\CIC}{\mathrm{CIC}}
\DeclareMathOperator{\IC}{\mathrm{IC}}

\newcommand{\coins}{\mathrm{coins}}
\newcommand{\braket}[2]{\langle #1|#2\rangle}
\newcommand{\ketbra}[2]{|#1\rangle\langle #2|}
\newcommand{\qotp}{ QOTP }

\newcommand{\XX}{\mathbf{X}}
\newcommand{\YY}{\mathbf{Y}}

\newcommand{\calU}{\mathcal{U}}

\newcommand{\NN}{\mathbb{N}}
\newcommand{\ZO}{\ensuremath{\{0,1\}}}

\newcommand{\AND}{\hbox{AND}}
\newcommand{\DISJ}{\hbox{DISJ}}
\newcommand{\OR}{\hbox{OR}}

\newtheorem{theorem}{Theorem}

\newtheorem{definition}[theorem]{Definition}
\newtheorem{remark}[theorem]{Remark}
\newtheorem{proposition}[theorem]{Proposition}
\newtheorem{lemma}[theorem]{Lemma}
\newtheorem{corollary}[theorem]{Corollary}
\newtheorem{claim}[theorem]{Claim}

\newcommand{\eps}{\varepsilon}

\newcommand{\zo}{\{0,1\}}

\newcommand{\ie}{{\it i.e.},~}

\def\memls{{\mathrm{ML}}}
\def\coins{{\mathrm{C}}}
\def\oscoins{{\mathrm{C_1}}}

\allowdisplaybreaks

\title{The information cost of quantum memoryless protocols}
\author{Andr{\'e} Chailloux\footnote{Inria, Paris, 
	\texttt{andre.chailloux@inria.fr}},
	Iordanis Kerenidis\footnote{CNRS IRIF, Universit{\'e} Paris 7, 
	\texttt{jkeren@liafa.univ-paris-diderot.fr}},
	Mathieu Lauri{\`e}re\footnote{NYU-ECNU Institute of Mathematical Sciences at NYU Shanghai,
	\texttt{mathieu.lauriere@nyu.edu}}}

\begin{document}

\maketitle

\begin{abstract}
We consider memoryless quantum communication protocols, where the two parties do not possess any memory besides their classical input and they take turns performing unitary operations on a pure quantum state that they exchange between them. Most known quantum protocols are of this type and recently a deep connection between memoryless protocols and Bell inequality violations has been explored in~\cite{BCG+16}. We study the information cost of memoryless quantum protocols by looking at a canonical problem: bounded-round quantum communication protocols for the one-bit AND function. We prove directly a tight lower bound of $\Theta(\frac{\log k}{k})$ for the information cost of AND for $k$-round memoryless quantum protocols and for the input distribution needed for the Disjointness function.  

It is not clear if memoryless protocols allow for a reduction between the AND function and Disjointness, due to the absence of private workspaces. We enhance the model by allowing the players to keep in their private classical workspace apart from their classical input {\em also} some classical private coins.
Surprisingly, we show that every quantum protocol can be transformed into an equivalent quantum protocol with private coins that is perfectly private, i.e. the players only learn the value of the function and nothing more. Last, we consider the model where the players are allowed to use one-shot coins, i.e. private coins that can be used only once during the protocol. While in the classical case, private coins and one-shot coins are equivalent, in the quantum case, we prove that they are not. More precisely, we show that every quantum memoryless protocol with one-bit inputs that uses one-shot coins can be transformed into a memoryless quantum protocol without private coins and without increasing too much its information cost. Hence, while private coins always allow for private quantum protocols, one-shot coins do not.
\end{abstract}

%
%
%
%




\section{Introduction}\label{mland:sec:intro}
\subsection{Context}

In the model of communication complexity, two players, Alice and Bob, receive inputs and would like to solve some distributed task that depends on these inputs, while minimizing the number of bits they exchange. This model has deep connections to many areas of computer science, including data structures, circuit lower bounds and streaming algorithms~\cite{MR1426129-KushilevitzN1997}.

Recently, a lot of attention has been given to a different measure of complexity for communication protocols, namely the amount of information that is leaked about the players inputs during the protocol. The information cost of a protocol is always lower than the communication cost, since one communicated bit can carry at most one bit of information. It has proved to be one of the strongest techniques we have to lower bound the communication complexity of functions~\cite{MR2059642-Bar-YossefJKS-2004,MR2933311-BravermanR-2010-amortized,MR3424073-Braverman-2012,KLL+12}. 

One can also define the notions of communication and information complexity in the quantum setting, where the two players exchange quantum messages. While it is straightforward to define the communication cost of a quantum protocol as the number of qubits that the two players exchange, one has to be careful when defining the information cost of a quantum protocol. 

Besides some application-specific definitions~\cite{JRS03,JN14}, recently two main definitions have been put forward. Touchette \cite{Tou15} has defined a notion of quantum information cost (QIC) and has proved that it has a number of important properties, including that for any function the quantum information complexity, namely the information cost of the optimal quantum protocol that solves the function, equals the amortized communication complexity of the function. Kerenidis et al. \cite{MR3497091-KerenidisLLGR} proposed a different notion, the classical input information cost (CIC), that is more intuitively related to the information leakage of the protocol, but is smaller than the QIC notion. Very recently, Lauri{\`e}re and Touchette \cite{LauriereT16} clarified the relation between the two notions showing that while CIC measures how much information each player learns about the other's input during the protocol, the QIC measures, on top of this, the information the players forget during the protocol. 

While our understanding of the flow of information during a quantum protocol has deepened, both these notions remain difficult to use in practice. The main reason is that mathematically they both involve a quantum conditional mutual information, where the conditioning is on a quantum variable. This quantity is notoriously difficult to handle, even though there has been some recent breakthrough work on it~\cite{MR3397027-FawziR-2015}.

Here, we try to overcome this difficulty by looking at a rich subclass of quantum protocols that we call \emph{memoryless} quantum protocols. In these protocols, the two parties take turns performing unitary operations on a pure quantum state that they exchange between them. In other words they do not possess any memory and hence they do not keep anything in their private space apart from their classical input. 

There are many reasons why it is interesting to look at such protocols. First, almost all quantum protocols we know are memoryless. This includes all protocols in the simultaneous message passing model, eg. fingerprints for Equality~\cite{buhrman2001quantum}, and in the one-way model, e.g. Hidden Matching~\cite{MR2399531-Bar-YossefJK-2008,MR2476272-GavinskyKKRdW-2008}, but also the two-way protocol for Disjointness in~\cite{BCW98} and for Vector in Subspace~\cite{raz1999exponential,KR11}. We note that the optimal protocol for Disjointness by Aaronson-Ambainis~\cite{MR2322514-AaronsonA-2005} is not described as a memoryless protocol.  
Second, there is a deep connection between memoryless protocols and Bell inequality violations that has been explored in~\cite{BCG+16,laplante2016robust}. Third, moving towards implementations of quantum communication protocols and the realization of quantum networks, memoryless protocols can be much easier to implement as it has already been shown~\cite{XAW+15,GXY+16}. Last but not least, we will argue it may be easier to understand the flow of quantum information in memoryless protocols. For example, it is easy to see that the relation between the two notions, CIC and QIC, is in this case clear: for any memoryless protocol, QIC is exactly two times the CIC, since the players forget exactly as much as they learn. 
Note that forgetting is not necessarily a drawback of quantum protocols: forgetting can be required in order to obtain quantum communication speed-ups, as shown in~\cite{LauriereT16}.

\subsection{Contributions}

In our work, we study the information cost of memoryless quantum protocols by looking at a canonical problem: bounded-round quantum communication for the AND function, where the players receive one bit each and their goal is compute the AND of the inputs. 

One of the main reasons to study the AND function is its close relation to the Disjointness problem (DISJ), where the players receive one set each and their goal is to decide whether these two sets are disjoint. One can see DISJ as a function that takes as inputs two $n$-bit strings $x,y$ and returns the OR of the coordinate-wise AND of these strings, i.e. $\DISJ(x,y)=\OR ( \AND(x_1, y_1),\ldots, \AND(x_n, y_n))$. 
In the classical world, a very elegant lower bound for Disjointness using information-theoretic tools was given by Bar-Yossef et al. \cite{MR2059642-Bar-YossefJKS-2004} and its proof consists of the following two steps: first, one can reduce DISJ to AND. Namely, given a protocol for DISJ on inputs of size $n$, the players can construct a protocol to solve AND in the following way: they embed their one-bit inputs for AND in some random coordinate for DISJ, use their private coins to pick the remaining $(n-1)$ inputs uniformly from $\{ (0,0), (0,1), (1,0) \}$, and run the DISJ protocol. The output of DISJ for such inputs is the same as the output of the AND function. One can show this way that if the information cost of the DISJ protocol is $I$, then the information cost of the new protocol for AND is $I/n$. This implies the information complexity of DISJ is at least $n$ times the information complexity of AND for the above input distribution. The second stage, involves computing directly the information complexity of the AND function, and showing to be at least a constant, the tight $\Omega(n)$ lower bound for DISJ is obtained. This result not only gives a simpler proof of the DISJ lower bound but it has sparked the interest for the study of information complexity that has led to numerous advances~\cite{MR2933311-BravermanR-2010-amortized,MR3424073-Braverman-2012,MR3210776-BravermanGPW-2013-exactcomm,MR3388213-BravermanW-2015-odometer}.

In the quantum world, things are considerably more complicated. The first attempt to provide an information-theoretic proof of the bounded-round quantum communication complexity of DISJ was by Jain et al. \cite{JRS03}. In their work, they used a different information-theoretic notion from QIC and CIC and used it to reduce the DISJ problem to the AND problem. By directly lower bounding this quantity for the AND function they managed to get a lower bound of $\Omega(n/k^2)$, for any $k$-round protocol for DISJ. There is no clear way to improve this lower bound using their information-theoretic notion and this bound falls short of the optimal bound of $\Omega(n/k)$. Very recently,~\cite{MR3473340-BravermanGKMT-2015} provided a proof which gives an almost optimal bound of $\tilde \Omega(n/k)$ for $k$-round protocols for DISJ by reducing DISJ to AND and then using the already known lower bound for DISJ to lower bound the complexity of AND. Note that this proof does not provide a direct proof for the information complexity of AND. 

In our work we prove directly a tight lower bound for the information complexity of AND for memoryless protocols and for the input distribution needed for DISJ.  
More precisely, considering the input distribution $\calU_0$ defined by $\calU_0(x,y) = \frac{1}{3}$ for $(x,y) \neq (1,1)$ and $\calU_0(1,1) = 0$, we show the following, where $\CIC^{\memls}_{\calU_0,\eps,k}(\AND)$ is the minimum $\CIC$ achieved by a $k$-round memoryless quantum protocol computing AND with error at most $\eps$ on input distribution $\calU_0$.
\begin{restatable}{theorem}{pureANDCIC}\label{Theorem:CIC_pure}
	For any $\eps \in (0,1/2)$ and any integer $k$,
	$ \CIC^{\memls}_{\calU_0,\eps,k}(\AND) = \Theta_\eps\left(\frac{\log(k)}{k}\right).$
\end{restatable}
The upper bound in Theorem~\ref{Theorem:CIC_pure} comes from a protocol described in \cite{MR3473340-BravermanGKMT-2015} credited to Jain, Radhakrishnan and Sen. Note also that from~\cite{JRS03}, we could obtain a non-optimal bound of 
$ \CIC^{\memls}_{\calU_0,\eps,k}(\AND) = \Omega_\eps\left(1/k\right),
$
since the information-theoretic notion used in~\cite{JRS03} becomes equivalent to CIC for memoryless protocols.

The question is then whether we can lift the lower bound of Theorem~\ref{Theorem:CIC_pure} to memoryless quantum protocols for DISJ. The obvious way to try and do it is to start with a memoryless quantum protocol for DISJ and use it in order to construct a memoryless protocol for AND. However, there is a problem. To solve AND, the players are given one-bit inputs, say $x$ and $y$. But if they want to use a protocol solving DISJ over $n$ bits, they need to create $n-1$ inputs for each party distributed in a way such that the protocol for DISJ will actually compute $\AND(x,y)$. In the classical case, the players used private coins to choose the remaining inputs for DISJ, when we embed the AND function to it. In \cite{JRS03}, the players used a superposition of coins in order to choose these inputs. Now, if the players keep these superpositions in their workspace, then we lose the memoryless property of our protocols. On the other hand, if they send these superpositions to the other player, the information cost of the protocol might considerably increase. 

Since it is not obvious how to reduce DISJ to AND while retaining the memoryless property, similarly to the classical case where we do not know how to perform the reduction without the use of private coins, we slightly enhance our model to try to allow for this reduction to go through.  

More precisely, we look at the model where the players do not possess any memory and hence they do not keep anything in their private space apart from their classical input {\em and} some classical private coins. Note that one can also assume the players share public coins without changing the model. 
In the classical case, we do allow for private coins when we define the information cost of a protocol. In the quantum case, note that we cannot unitarily create classical coins. Allowing classical coins seems like a minimal addition to the model. One can see that the communication complexity in this new model is not different from the communication complexity in the model without coins. Indeed, any protocol with coins can be simulated by a protocol where the coins are created in superposition by the players without changing the communication cost.
But what about the information complexity? One one hand, the information complexity cannot increase, since we can always ignore the coins. Surprisingly, we show that it becomes as small as it can possibly be, namely, it equals the information revealed just by the value of the function. In other words, we say that any function can be computed privately.
In fact, we show that every quantum protocol can be turned into a quantum protocol with coins that has the same input-output behaviour as the original protocol and that is perfectly private, i.e. the players only learn the value of the function the protocol computes and nothing more.  


\begin{restatable}{theorem}{zeroCICPif}\label{thm:zeroCIC-Pif}
	For every quantum communication protocol $\Pi$, there exists a memoryless quantum protocol $\Pi'$ with private classical coins such that:
	\begin{itemize}
		\item on every input pair $(x,y)$, $\Pi'$ has the same output distribution as $\Pi$.
		\item the information cost of $\Pi'$ is only the information gained by Bob's output $\Pi_{out}$ in $\Pi$. This means that for every input distribution $\mu$, we have
		$\CIC_\mu(\Pi') = 
		I\left(\Pi_{out}(X,Y) :X|Y\right),
		$
	\end{itemize}
	where $(X,Y)$ is a random variable distributed according to $\mu$, $I$ denotes the (classical) conditional mutual information and $\Pi_{out}(x,y)$ is the (classical) random variable corresponding to Bob's output in $\Pi$ on input $(x,y)$. 
\end{restatable}
Although we call the protocol $\Pi'$ private, note that we are not interested in a cryptographic scenario where the players might deviate from the protocol: we are interested in studying the information of fixed protocols. 
In high level, in protocol $\Pi'$ Alice and Bob follow $\Pi$ but they use private coins to encrypt their messages. At the end of the protocol, if their coins were the same, Bob is able to output as in $\Pi$ and knows nothing else than this value. However, if their coins were different, our construction prevents them from getting any information about each other's input, in which case they just restart the process until they get the same coins. This construction yields a private protocol at the expense of a very high communication cost. 

Our results imply that given any function $f$, if we take for $\Pi$ the protocol where Alice just sends over $x$ to Bob who computes and output $f(x,y)$, we obtain a protocol $\Pi'$ that can perfectly compute $f$ with $\CIC$ only the information gained from $f(x,y)$. 

\begin{restatable}{corollary}{zeroCICPifcorofct}\label{thm:zeroCIC-Pif-corofct}
	For every input distribution $\mu$, and every positive integer $k$,
		$\CIC^{\memls,\coins}_{\mu,k}(f) = I\left( f(X,Y) :X|Y\right),$
	where $(X,Y)$ is a random variable distributed according to $\mu$, $I$ denotes the (classical) conditional mutual information, and $\CIC^{\memls,\coins}_{\mu,k}(f)$ denotes the minimum $\CIC$ achieved by a $k$-round memoryless quantum protocol with private classical coins to compute $f$ exactly on input distribution $\mu$.
\end{restatable}
Note that, in the case of AND, on distribution $\calU_0$ the output of $\AND(x,y) = x \wedge y$ is always $0$. Hence, by the above result, $\CIC^{\memls,\coins}_{\calU_0,k}(\AND) = 0$ for every integer $k$.

There are two sides to this result. On the one hand, adding classical coins to quantum protocols allows for perfectly private protocols. This is impossible in the classical world and shows how quantum communication can offer advantages over classical communication. On the other hand, allowing the players to use private coins without restrictions weakens the power of information complexity as a lower bound for quantum protocols. 

In order to try and salvage the notion of information complexity as a strong lower bound while allowing the players to use private coins, we consider 
an intermediate model where the players are allowed to use what we call \emph{one-shot coins}. These are private coins that can be used only once during the protocol (then the players forget them). In the classical setting, this assumption is not restrictive and does not change the communication complexity nor the information complexity: for every protocol with private coins, we can construct a protocol which has the same transcript and output behaviours and uses only one-shot coins. This construction is done by changing only the way the coins are used~\cite{LauriereT16}. 

In the quantum setting, this is not the case:  allowing general coins or one-shot coins can lead to very different information complexities. We, in fact, show that one-shot coins do not necessarily decrease the information cost of a general quantum protocol by much. 
More precisely, we denote by $\CIC^{\memls,\oscoins}_{\mu,\epsilon, k}(f) $ the minimum $\CIC$ achieved by a $k$-round memoryless quantum protocol with private one-shot coins that computes $f$ with error $\epsilon$ on input distribution $\mu$, and prove that

\begin{restatable}{theorem}{oneshotcoinsBound}\label{mland:thm:oneshotres}
	For every $k$-round memoryless quantum protocol $\Pi$ with one-shot coins, we can construct a $k$-round memoryless protocol $\Pi'$ without coins, which has the same behaviour distribution as $\Pi$ and such that 
	$
	\CIC_{\calU_0}(\Pi') 
	=
	O\big(\CIC_{\calU_0}(\Pi) \cdot \big( \log(k) +  \left| \log \CIC_{\calU_0}(\Pi)\right| \big) \big).
	$
\end{restatable}

Informally, the proof will go as follows. We will transform a memoryless quantum protocol $\Pi$ with one-shot coins into a memoryless quantum protocol $\Pi'$ without any coins, such that $\Pi'$ will have the same outputs than $\Pi$ and without increasing the information cost too much. The transformation from $\Pi$ to $\Pi'$ will informally be the following:
\begin{enumerate}
	\item Quantize the coins from $\Pi$ \ie put them in quantum superposition in quantum registers.
	\item At each odd (or even) round, Alice (or Bob) applies the same transformation as in $\Pi$. Then, Alice (or Bob) would like to send all their quantum registers, including the coin registers, to the other player. Before doing that, Alice (or Bob) applies a compensation unitary that will limit the information cost increase that occurs because of the sending of all the quantum registers. 
\end{enumerate}

This result, combined with Theorem~\ref{Theorem:CIC_pure}, implies in particular that 
$ \CIC^{\memls, \oscoins}_{\calU_0,\eps,k}(AND) = \Theta_\eps\left(\frac{1}{k}\right).$ We see that while private coins allow for private protocols, one-shot coins not always do. The main open question is whether one-shot coins can be useful to reduce DISJ to AND or more generally prove some direct sum property for quantum information complexity.


\section{Preliminaries}

%
%
%
\subsection{Quantum Communication}

In the sequel we consider quantum protocols with classical inputs as defined in Figure~\ref{fig:qu_prot_CIC}. At the end of the protocol, Bob applies a measurement (POVM) on register $\mathcal O$ to obtain a classical output.

Let $k$ be an odd positive integer. A $k$-round protocol can be described as follows. At the outset of the protocol, Alice and Bob receive a classical input, in registers $X$ and $Y$ respectively. They will not modify the content of these registers. Alice starts with empty private memory and message registers, $A_0$  and $M_0$ respectively; Bob starts with an empty private memory register $B_0$.
At each odd round $1 \leq i \leq k$, Alice applies a unitary operation on her private memory register, $A_i$, and the message register, $M_i$. This operation is controlled by her (classical) input. She then sends $M_i$ to Bob. At each even round $i$, Bob applies a unitary operation, controlled by his input, on his private memory register, $B_i$, and the message register, $M_i$, and then sends $M_i$ to Alice. After round $k$ (Bob has just received message $M_k$ from Alice since $k$ is odd), Bob applies a unitary $U_{k+1}$ and measures part of his registers to obtain a classical output. For simplicity, let us denote
$
	B_{k+1} = U_{k+1}(B_kM_k).
$
We will split $B_{k+1}$ into registers $\mathcal G_B$ and $\mathcal O$  such that the part of $B_{k+1}$ that is measured by Bob at the end is $\mathcal O$.

\begin{center}
\scalebox{.8}{%
\begin{overpic}[width=1\textwidth]{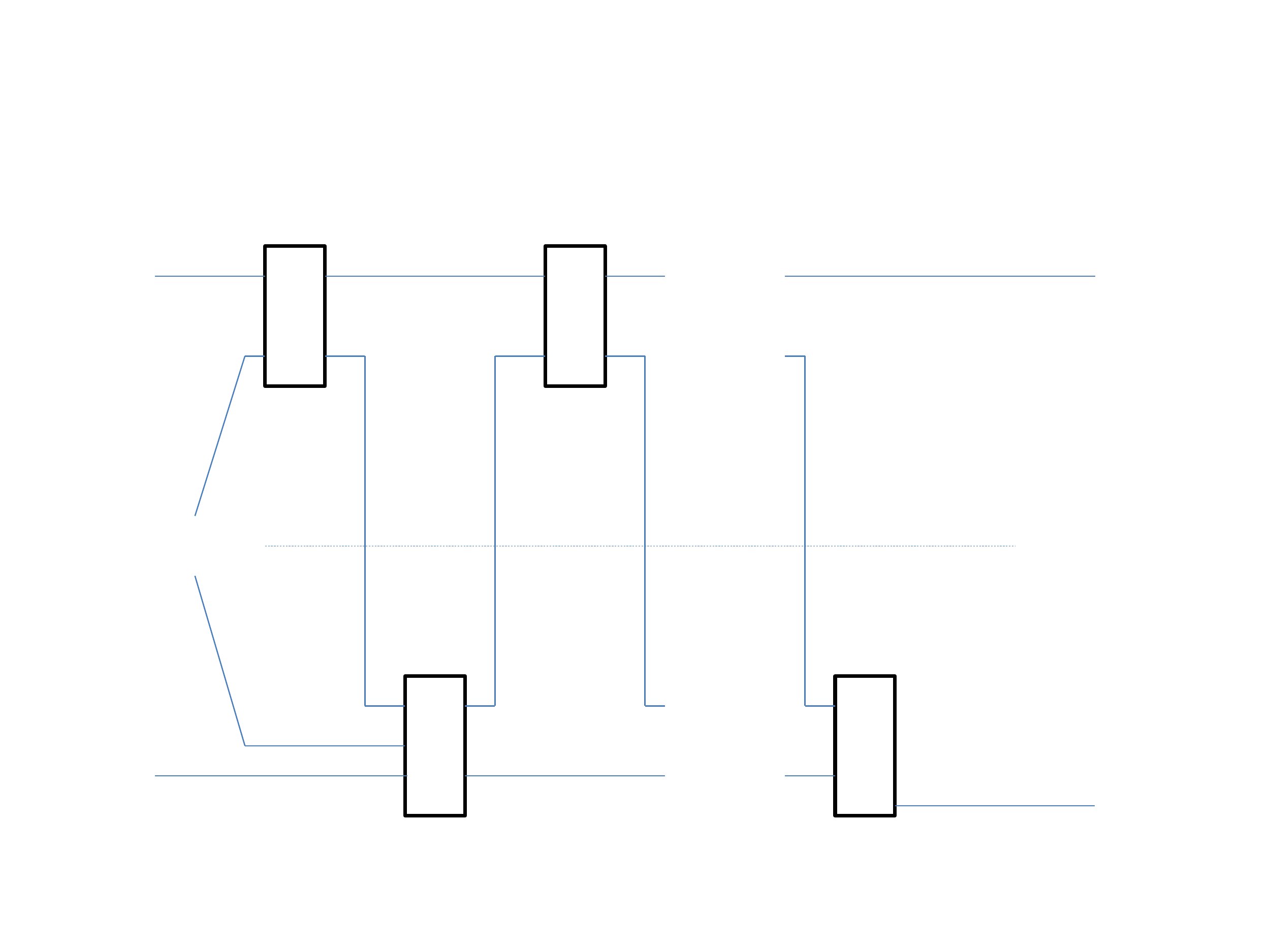}
  \put(0,45){Alice}
  \put(0,22){Bob}
  \put(15,54){\footnotesize$X$}
  \put(15,11.5){\footnotesize$Y$}
  \put(22.1,49.5){\footnotesize$U_1$}
  \put(15,31.5){\footnotesize$\ket{0}$}
  \put(27,54){\footnotesize$X A_1$}
  \put(33.5,54){\footnotesize$=$}
  \put(26.2,48){\footnotesize$M_1$}
  \put(33.5,15.5){\footnotesize$U_2$}
  \put(37,54){\footnotesize$X A_2$}
  \put(37.2,17.4){\footnotesize$M_2$}
  \put(37.2,11.5){\footnotesize$Y B_2$}
  \put(44,11.5){\footnotesize$=$}
  \put(44.2,49.5){\footnotesize$U_3$}
  \put(48.2,54){\footnotesize$X A_3$}
  \put(48.2,48){\footnotesize$M_3$}
  \put(48.2,11.5){\footnotesize$Y B_3$}
  \put(55,33){\footnotesize$\cdots$}
  \put(59.5,54){\footnotesize$X A_{k}$}
  \put(66,54){\footnotesize$=$}
  \put(59.5,48){\footnotesize$M_{k}$}
  \put(58,11.5){\footnotesize$Y B_{k}$}
  \put(66,15.5){\footnotesize$U_{k+1}$}
  \put(70,54){\footnotesize$X A_{k+1}$}
  \put(73.5,9.5){\footnotesize$Y B_{k+1} = Y \mathcal G_B \mathcal O$}
\end{overpic}
}
  \captionof{figure}{
  \footnotesize
Depiction of a quantum protocol with classical inputs in the interactive model, adapted from the long version of~\cite[Figure 1]{Tou15}.
The local operations $U_i$ are unitary operators. The $A_i$ registers are Alice's local memory in round $i$, similarly the $B_i$'s are Bob's, the $M_i$'s are the message registers, exchanged back-and-forth between them. $X$ and $Y$ are Alice's and Bob's respective input, distributed according to some prior $\mu$.
}
  \label{fig:qu_prot_CIC}
\end{center}

\begin{definition}[Quantum communication complexity]
\label{def:QCC}
The quantum communication cost of a quantum protocol $\Pi$ is:
$
	\QCC(\Pi) = \sum_{i } \log \dim(M_i),
$
 where $M_i$ is the quantum register that corresponds to the message sent over at round $i$. 

For a fixed function $f$, the average-case quantum communication complexity of $f$ relative to input distribution $\mu$ and with (distributional) error $\epsilon \in [0,1]$ is defined as
	\begin{equation*}
		\QCC_{\mu,\epsilon}(f) = \inf_{ \Pi} \QCC(\Pi)
	\end{equation*}
	where the infimum is taken over all protocols $ \Pi$ computing $f$ with average error at most $\epsilon$ with respect to input distribution $\mu$. The worst-case quantum communication complexity of $f$ with error $\epsilon$ is defined as
	\begin{equation*}
		\QCC_{\epsilon}(f) = \inf_{\Pi} \QCC(\Pi)
	\end{equation*}
	where the infimum is taken over all protocols $ \Pi$ computing $f$ with worst-case error at most $\epsilon$.
\end{definition}

\subsection{The Classical input Information Cost ($\CIC$) of quantum protocols}

Let us recall the usual definition of information cost in classical protocols (see e.g.~\cite{CSWY01,MR3424073-Braverman-2012}). 
The information cost of a  protocol $\Pi$ is the amount information revealed by the whole transcript of $\Pi$, that is, the messages and the public coins. For simplicity, we also denote $\Pi$ the random variable corresponding to the transcript of $\Pi$ on input $X,Y$. Then we have
\begin{equation*}
	\IC(\Pi) = I(X:\Pi | Y, R_B) + I(Y:\Pi | X, R_A).
\end{equation*}
By a chain rule argument, it is possible to rewrite IC as the sum over the rounds of the information that the receiver of a message learns about the input of the other player, conditioned on what he already knew before receiving the message (that is, the information contained in his input and the past messages).
\begin{definition}[Information cost]
\label{def:IC}
The information cost of a (randomized) classical protocol $\Pi$, over input distribution $\mu$, is
\begin{equation*}
	\IC(\Pi) = \sum_{i\,:\,\text{odd}} I(M_i:X|Y, R_B, M_{[i-1]}) + \sum_{i\,:\,\text{even}} I(M_i:Y|X, R_A, M_{[i-1]}),
\end{equation*}
where $M_{[i-1]} = M_1, \dots, M_{i-1}$ is the concatenation of all messages sent before round $i$.

For a fixed function $f$ the information complexity of $f$ relative to input distribution $\mu$ and with (distributional) error $\epsilon \in [0,1]$ is defined as
		$\IC_{\mu,\epsilon}(f) = \inf_{ \Pi } \IC_{\mu}(\Pi)$
	where the infimum is taken over all protocols $ \Pi$ computing $f$ with error at most $\epsilon$ with respect to $\mu$.
\end{definition}

Imitating this round-by-round definition of IC,~\cite{MR3497091-KerenidisLLGR} introduced a version of information cost for quantum protocols with classical inputs. 
It is the sum of the information that the message receiver learns about the sender's (classical) input at each round, conditioned on the registers held by the receiver. 

In the sequel, for a quantum state $\rho^{A,B,C}$ over registers $A,B$ and $C$ (i.e. $A,B,C$ are three disjoint finite dimensional quantum systems having some joint density matrix $\rho$), we denote by 
$\rho_{A} = \Tr_{BC}(\rho)$ the reduced density matrix of $A$ and by $S(A) = -\Tr(\rho_A \log(\rho_A))$ the von Neumann entropy of $A$. $I(A:B) = S(A) + S(B) - S(AB)$ denotes the quantum mutual information between $A$ and $B$ and 
$I(A:B|C) = I(A: B,C)-I(A: C)$ is the quantum condition mutual information between registers $A$ and $B$ conditioned on $C$.

With the notations introduced above for quantum protocols, this leads to the following definition. 
\begin{definition}[Classical input information complexity~\cite{MR3497091-KerenidisLLGR}]
\label{def:CIC}
The classical input information cost of a quantum protocol $\Pi$ over input distribution $\mu$ is:
\begin{equation*}
	\CIC_\mu(\Pi) = \sum_{i\,:\,\text{odd}} I(M_i:X|Y, B_i) + \sum_{i\,:\,\text{even}} I(M_i:Y|X, A_i).
\end{equation*}

For a fixed function $f$ the classical input quantum information complexity of $f$ relative to input distribution $\mu$ and with (distributional) error $\epsilon \in [0,1]$ is defined as
	\begin{equation*}
		\CIC_{\mu,\epsilon}(f) = \inf_{ \Pi} \CIC_{\mu}(\Pi)
	\end{equation*}
	where the infimum is taken over all protocols $ \Pi$ computing $f$ with error at most $\epsilon$ with respect to $\mu$.
\end{definition}
Touchette \cite{Tou15} has defined a different notion of quantum information cost ($\QIC$), but as shown in~\cite{LauriereT16}, $\QIC$ and $\CIC$ are equivalent within a factor of 2. In fact, for memoryless protocols $\QIC$ is exactly two times $\CIC$, and hence we only look at $\CIC$ in the rest of the paper.


\subsection{Measures of distance between quantum states}

Let us recall a few tools about classical-quantum states and measures of distance between two quantum states. If $C$ is a classical random variable taking the classical value $\ket c$ with probability $p_c$,  then $I(A : B | C) = \sum_c p_c \, I(A^c : B^c)$, where $(AB)^c$ denotes the joint density matrix of $A$ and $B$ when $C=\ket c$. We also write $I(A:B|C=c)$ for $I(A^c :B^c)$.
\begin{definition}
Let $\rho$ and $\sigma$ be density matrices in the same finite dimensional Hilbert space. The trace distance between $\rho$ and $\sigma$ is defined as
	$
	\Delta(\rho, \sigma) = \frac 12 \Tr\left(\sqrt{(\rho-\sigma)^\dagger (\rho-\sigma)}\right),
	$
where $\dagger$ denotes the conjugate-transposition.
\end{definition}
We recall that $\Delta$ is preserved under unitary transformations: for every unitary operator $U$, and every $\rho$ and $\sigma$ as above,
$
	\Delta( U \rho U^\dagger, U \sigma U^\dagger) = \Delta(\rho, \sigma).
$
When the states are pure, say $\rho = \ketbra\psi\psi$ and $\sigma = \ketbra\phi\phi$, then
\begin{equation}\label{eq:TrdistPure}
	\Delta(\rho,\sigma) = \frac 12 \sqrt{1 - |\braket \psi\phi|^2}.
\end{equation}
\begin{definition}
	Let $\rho$ and $\sigma$ be density matrices in the same finite dimensional Hilbert space $H$. The fidelity (also called Bhattacharyya coefficient) between $\rho$ and $\sigma$ is defined as
 $F(\rho,\sigma) = \Tr\left( \sqrt{\sqrt{\rho} \, \sigma \sqrt{\rho}} \right)$.
\end{definition}
By Uhlmann's theorem, we also have the following characterization: $F(\rho,\sigma) = \sup_{K, \ket \phi, \ket \psi} |\braket\phi\psi|$, where the maximization is over Hilbert space $K$ and purification $\ket\phi$ and $\ket\psi$ in $H \otimes K$ of $\rho$ and $\sigma$ respectively.
We will make use of the following result.
\begin{lemma}[see e.g. \cite{NC00}]\label{fact:measureH}
%
	Let $\rho$ and $\sigma$ be density matrices in the same finite dimensional Hilbert space $H$. Let $E = \{E_m\}_m$ be a POVM on $H$. Let $p_m = \Tr(\rho E_m)$ and $q_m = \Tr(\sigma E_m)$. We have 
	$
	\Delta(\rho, \sigma) \ge \frac{1}{2}\sum_m |p_m - q_m|.
	$
\end{lemma}

\subsection{Memoryless quantum communication}

A particular class of quantum communication protocols are those where the players send back and forth the message register without keeping any information in their private memory register apart from their classical input.
\begin{definition}
A quantum communication protocol is said to be a memoryless protocol if and only if the following condition is satisfied: the players' private registers, namely $A_i,B_i$ (see Fig. \ref{fig:qu_prot_CIC}), are always empty, meaning that their working registers correspond only to the message registers. This implies that at each round, the sender will always send all the registers he or she currently has except the input register to the other player.
\end{definition}
While at first hand, this family of protocols seems restricted, since messages can not be entangled with private memory registers (nor with the environment), they constitute a powerful and worth-studying class of protocols. 
As we have said in Section~\ref{mland:sec:intro}, almost all quantum protocols we know are memoryless. Moreover, there is a deep connection between memoryless protocols and Bell inequality violations and such protocols can be much easier to implement. Last but not least, we will see in the following sections that it is easier to understand the flow of quantum information in memoryless protocols, for example, the information that is learnt in the protocol is exactly the same as the information been forgotten.
 
 Let us denote by ${\mathcal T}_{\mu,\epsilon,k}^{\memls}(f)$ the set of all $k$-round memoryless protocols $ \Pi$ computing function $f$ with error at most $\epsilon$ with respect to $\mu$.
\begin{definition}\label{def:QCC1}
	The $k$-round $\epsilon$-error memoryless quantum communication complexity of a function $f$ relative to input distribution $\mu$ is defined as
		$
		\QCC_{\mu,\epsilon,k}^{\memls}(f) = \inf_{ \Pi \in {\mathcal T}_{\mu,\epsilon,k}^{\memls}(f)} \QCC_{\mu}(\Pi).
		$
\end{definition}

\begin{definition}\label{def:QIC1}
	The $k$-round $\epsilon$-error memoryless classical input information complexity of $f$ relative to input distribution $\mu$ is defined as
		$
		\CIC_{\mu,\epsilon,k}^{\memls}(f) = \inf_{ \Pi \in {\mathcal T}_{\mu,\epsilon,k}^{\memls}(f)} \CIC_{\mu}(\Pi).
		$
\end{definition}

In \cite{JRS03}, the authors used a notion of quantum information loss  (denoted QIL here) of a protocol $\Pi$, defined, with the notation introduced above, by:
$$
	\QIL_{\mu,\epsilon,k}(\Pi) = \sum_{i \,:\, \mathrm{odd}} I(X : M_i B_i |Y) + \sum_{i \,:\, \mathrm{even}} I(Y : M_i A_i |X). 
$$
For memoryless protocols, private memory registers $A$ and $B$ are empty, so this notion coincides with the notion of classical input information cost ($\CIC$).


\section{The CIC of bounded-round memoryless protocols for AND}\label{mland:sec:andpure}

In this section, we prove the following result, where $\calU_0$ denotes the uniform distribution over $\{(0,0), (0,1), (1,0)\}$.

\pureANDCIC*

This is a consequence of our lower bound in Proposition~\ref{prop:ICeps} and the upper bound in \cite{JRS03,MR3473340-BravermanGKMT-2015} (see Section~\ref{qvariants-lbtight}).

\subsection{Optimal lower bound for the information cost of a pure message}

We first study the information cost of a single quantum message which is in a pure state.
Let us recall that, using Pinsker's inequality, we have the following result.
\begin{lemma}[see~\cite{MR1686254-FuchsG-1999}]
\label{lem:mi-fid}
Suppose $X$ and $Q$ are disjoint quantum systems, where $X$ is a classical random variable uniformly distributed over $\ZO$ and $Q$ is a quantum encoding $x \rightarrow \sigma_x$ of $X$. Then, the following relationship holds between mutual information and fidelity: 
$$I(X : Q) \geq 1-F(\sigma_0, \sigma_1).$$
\end{lemma}
In the sequel, for a quantum state $Q$ depending on classical inputs in $\XX \times \YY$, we denote by $Q^{x,y}$ the value of $Q$ when the inputs are $(x,y)$. The above result leads to the following corollary, which applies to the general situation where $M$ is a quantum message which is non necessarily in a pure state.
\begin{corollary}\label{mland:coro:mutualInfoLB-usual}
	Let $X,Y$ be two classical random variables taking value in $\zo$. Let $M$ be a quantum encoding of $(X,Y)$. Then
$
	I(X : M | Y=y) \geq 1 - F(M^{1y}, M^{0y}).
$
\end{corollary}
In this section, we prove a tighter inequality (i.e. a larger lower bound) that holds when $M$ is a pure state. Let $h_2$ be the binary entropy function:
$
	h_2 : [0,1] \to [0,1], \, p \mapsto -p \log p - (1-p) \log (1-p)
$
 where $\log$ denotes the base $2$ logarithm, with the convention that $0 \log 0 = 0$, and let $h_2^{-1}$ be its inverse function from $[0,1]$ to $[0, 0.5]$.

\begin{lemma}\label{Lemma:PurePinsker}
Consider the bipartite quantum state over two registers, $X$ and $M$, 
$
	\rho = \sum_{x \in \zo} \frac{1}{2} \ket{x}\bra{x}_X \otimes \ket{\psi_{x}}\bra{\psi_{x}}_M,
$
for some pure states $\ket{\psi_{x}}$, $x = 0,1$.
Then 
$$
	I(X:M) = h_2\left(\frac{1 - |\braket{\psi_{0}}{\psi_1}|}{2}\right).
$$
\end{lemma}
\begin{proof}
Let $\delta \in [0,1]$ and $\alpha \in (-\pi/2, \pi/2]$ such that $\cos(\alpha) = 1 - \delta = |\braket{\psi_{0}}{\psi_1}|$.
We have
\begin{align*}
	& I(X:M)
	= S(X) + S(M) - S(XM)
	\\
	= \, & 1 + S\left(\frac{1}{2} \ketbra{\psi_{0}}{\psi_{0}} + \frac{1}{2} \ketbra{\psi_{1}}{\psi_{1}}\right) - S\left(\frac{1}{2} \ketbra{0}{0} \otimes \ketbra{\psi_{0}}{\psi_{0}} + \frac{1}{2} \ketbra{1}{1} \otimes \ketbra{\psi_{1}}{\psi_{1}}\right)
	\\
	= \, & S\left(\frac{1}{2} \ketbra{\psi_{0}}{\psi_{0}} + \frac{1}{2} \ketbra{\psi_{1}}{\psi_{1}}\right).
\end{align*}
Let $N,N'>0$ such that $\ket{\psi_+} = \frac{1}{N} (\ket{\psi_0} + \ket{\psi_1})$ and $\ket{\psi_-} = \frac{1}{N'} (\ket{\psi_0} - \ket{\psi_1})$ are of norm $1$. We have 
\begin{align*}
	\frac{1}{2} \ketbra{\psi_{0}}{\psi_{0}} + \frac{1}{2} \ketbra{\psi_{1}}{\psi_{1}}
	= &\, \cos^2\left(\tfrac \alpha 2\right) \ketbra{\psi_+}{\psi_+} + \sin^2\left(\tfrac \alpha 2\right)\ketbra{\psi_-}{\psi_-}
	\\
	= &\, \left(1 - \frac{\delta}{2}\right) \ketbra{\psi_+}{\psi_+} + \frac{\delta}{2} \ketbra{\psi_-}{\psi_-}.
\end{align*}
Since $\ket{\psi_+}$ and $\ket{\psi_-}$ are orthogonal, this gives~:
\begin{align*}
	I(X:M)
	= & \, S\left(\frac{1}{2} \ketbra{\psi_{0}}{\psi_{0}} + \frac{1}{2} \ketbra{\psi_{1}}{\psi_{1}}\right)
	= S\left(\left(1 - \frac{\delta}{2}\right) \ketbra{\psi_+}{\psi_+} + \frac{\delta}{2} \ketbra{\psi_-}{\psi_-}\right)
	= h_2(\delta/2).
\end{align*}
\end{proof}

\begin{corollary}\label{lem:improvedTraceIC}
	Consider the bipartite quantum state $\rho = \sum_{x \in \zo} \frac{1}{2} \ketbra{x}{x}_X \otimes \ketbra{\psi_{x}}{\psi_{x}}_M$. We have 
		$I(X:M) \ge h_2\left(\frac{1}{4}\Delta^2(\ket{\psi_0},\ket{\psi_1})\right).$
\end{corollary}
\begin{proof}
From Lemma~\ref{Lemma:PurePinsker}, we have $I(X:M) \ge h_2\left(\frac{1 - |\braket{\psi_0}{\psi_1}|}{2}\right)$. The conclusion holds since 
$(1 - |\braket{\psi_0}{\psi_1}|) \ge \frac{1}{2}\Delta^2(\ket{\psi_0},\ket{\psi_1})$ by \eqref{eq:TrdistPure}, and $h_2$ is increasing on $[0,0.5]$.
\end{proof}
As a consequence, when the state is tripartite $\rho = \sum_{x,y \in \zo} \frac{1}{4} \ketbra{x}{x}_X \otimes \ketbra{\psi_{x,y}}{\psi_{x,y}}_M \otimes \ketbra{y}{y}_Y$,
we obtain
	$$
	I(X : M | Y=y) \geq  h_2\left(\frac{1}{4}\Delta^2(\ket{\psi_{0,y}},\ket{\psi_{1,y}})\right),
	$$
instead of the bound given by Corollary~\ref{mland:coro:mutualInfoLB-usual}.

\subsection{Optimal lower bound for the information cost of bounded-round memoryless quantum protocols for $\AND$}


We show the following result.
The proof is very similar to the one presented in~\cite{JRS03} and the improvement we obtain relies on Corollary~\ref{lem:improvedTraceIC}.

\begin{proposition}\label{prop:ICeps}
		$
		\CIC^{\memls}_{\calU_0,\epsilon,k}(\AND) \ge \frac{1}{12 k} (1-2\eps)^2 \log(k) = \Omega(\log(k)/k).
		$
\end{proposition}

In the proof, it will be easier to study the following quantity, defined for every $k$-round memoryless quantum protocol $\overline \Pi$~:
$$
	\CIC^{0}_{\mu}(\overline \Pi) = \sum_{i = 1, \, \mathrm{ odd}}^k I(M_i : X | Y=0) + \sum_{i = 1, \, \mathrm{ even}}^k I(M_i : Y | X=0).
$$
It is tightly related to the classical input information cost under distribution $\calU_0$.
\begin{lemma}\label{mland:lem:CICQIL0}
	$
	\CIC^{\memls}_{\calU_0,\epsilon,k}(\AND) = \inf_{\overline \Pi} \CIC_{\calU_0}(\overline\Pi) = \frac{2}{3} \inf_{\overline \Pi} \CIC^0_{\calU_0}(\overline \Pi),
	$
where  the infimum is over memoryless $k$-round protocols computing $\AND$ with error at most $\eps$ with respect to $\calU_0$.
\end{lemma}
\begin{proof}
Let us consider such a protocol $\overline \Pi$ and denote $M_i$ its $i$-th message ($1\leq i \leq k$).
We have 
\begin{align*}
\CIC_{\calU_0}(\overline\Pi)
& = \sum_{i\,:\,\text{odd}} I(M_i:X|Y) + \sum_{i\,:\,\text{even}} I(M_i:Y|X).
\end{align*}
Moreover,  under distribution $\calU_0$, for any odd $i \leq k$,
\begin{align*}
	I(M_i : X | Y)
	& = \PP(Y=0) I(M_i : X | Y=0) + \PP(Y=1) I(M_i : X | Y=1) 
	\\
	& = \frac{2}{3} I(M_i : X | Y=0),
\end{align*}
since, when $Y=1$, $X=0$ with probability $1$ so $I(M_i : X | Y=1) = 0$. Similarly, for any even $i$, 
$$
	I(M_i : Y | X) = \frac{2}{3} I(M_i : Y | X=0).
$$
From this, we conclude the lemma.
\end{proof}
We now prove Proposition~\ref{prop:ICeps}.
\begin{proof}[Proof of Proposition~\ref{prop:ICeps}]
By Lemma~\ref{mland:lem:CICQIL0}, it suffices to show that for every $k$-round memoryless protocol $\overline \Pi$,
$$
	\CIC_{\calU_0}^0(\overline \Pi) \geq \frac{(1-2\eps)^2 \log(k)}{8 k}.
$$
Consider such a memoryless protocol $\overline \Pi$. It can be described with the following notations.
Alice and Bob have respective inputs $x$ and $y$. Alice starts with $\ket{0} = \ket{\psi^{0}_{xy}}$. At odd (resp. even) round $i$, Alice (resp. Bob) sends a message.
For each $i \leq k$, let us denote $\ket{\psi_{x,y}^i}$ the state of the message register just after it has been sent over.
We assume that the last message is sent by Alice to Bob (\ie $k$ is odd).

For simplicity, let us define, for $i \leq k$,
\begin{equation*}
	a_i = \begin{cases}
			I(M_i : X | Y=0) &\hbox{ if $i$ is odd}, \\
			I(M_i : Y | X=0) &\hbox{ if $i$ is even}.
		\end{cases} 
\end{equation*}
So $\CIC_{\calU_0}^0(\overline \Pi) = \sum_{i=1}^k a_i$. 
For $i \leq k$, we also define 
\begin{equation*}
	b_i = \begin{cases}
			\Delta(\ket{\psi_{10}^i}, \ket{\psi_{00}^i}) &\hbox{ if $i$ is odd}, \\
			\Delta(\ket{\psi_{01}^i}, \ket{\psi_{00}^i}) &\hbox{ if $i$ is even},
		\end{cases} 
\end{equation*}
and
\begin{equation*}
	\delta_i = \begin{cases}
			\Delta(\ket{\psi_{01}^i},\ket{\psi_{11}^i}) &\hbox{ if $i$ is odd}, \\
			\Delta(\ket{\psi_{10}^i},\ket{\psi_{11}^i}) &\hbox{ if $i$ is even}.
		\end{cases} 
\end{equation*}

We first prove three claims.
\begin{claim}
\label{claim:deltakeps}
It holds that $\delta_k \ge 1 - 2\eps.$
\end{claim}
\begin{proof}
For any inputs $(x,y)$, Bob outputs $x \wedge y$ with probability at least $1-\eps$ at the end of the protocol, where he has the state $\ket{\psi_{xy}^k}$. In particular, on inputs $(1,0)$, Bob will output $0$ with probability $1-\eps$ and on inputs $(1,1)$, Bob will output $1$ with probability $1-\eps$. The conclusion holds by Lemma~\ref{fact:measureH} and by the fact that $\delta_k = \Delta(\ketn{\psi_{01}^k},\ketn{\psi_{11}^k})$.
\end{proof}

\begin{claim}
\label{claim:deltakbi}
It holds that $\delta_k \le 2\sum_{i=1}^k b_i. $
\end{claim}
\begin{proof}
Let us first notice that for every $i \leq k$, $\delta_{i} \leq b_{i-1} + b_{i} + \delta_{i-1}$. Indeed, if $i$ is odd,
\begin{align*}
	\delta_{i} &= \Delta(\ket{\psi_{01}^i},\ket{\psi_{11}^i}) \\
			&\leq\Delta(\ket{\psi_{01}^i},\ket{\psi_{00}^i})
			 + \Delta(\ket{\psi_{00}^i},\ket{\psi_{10}^i})
			 + \Delta(\ket{\psi_{10}^i},\ket{\psi_{11}^i}) \tag{by triangle inequality}\\
			&\leq \Delta(\ket{\psi_{01}^{i-1}},\ket{\psi_{00}^{i-1}})
			+ b_{i}
			+ \Delta(\ket{\psi_{10}^{i-1}},\ket{\psi_{11}^{i-1}}) \tag{by unitary invariance of $\Delta$}\\
			&= b_{i-1} 
			+ b_{i}
			 + \delta_{i-1},
\end{align*}
where the second inequality holds since, at round $i$, Alice is the sender, and the unitary she applies depend only on her input (which is $x=0$ for both $\ket{\psi_{01}^i}$ and $\ket{\psi_{00}^i}$, and $x=1$ for both $\ket{\psi_{10}^i}$ and $\ket{\psi_{11}^i}$).
We repeat a similar argument if $i$ is even. Then, Claim~\ref{claim:deltakbi} is obtained by repeating this in a recursive way down to $i=1$. 
\end{proof}

\begin{claim}
\label{claim:biQIL}
	If for all $i \leq k$ we have $a_i \le 0.4$, then $\sum_{i=1}^{k} b_i \le 2k \sqrt{h_2^{-1}\left(\frac{\CIC_{\calU_0}^0(\overline \Pi)}{k}\right)}.$
\end{claim}
\begin{proof}
Let us assume that for all $i \leq k$ we have $a_i \le 0.4$.
	By Corollary~\ref{lem:improvedTraceIC}, we have for all $i \leq k$ that $b_i \le 2 \sqrt{h_2^{-1}(a_i)}$, hence
	$$
		\sum_{i=1}^k b_i \le \sum_{i=1}^k 2 \sqrt{h_2^{-1}(a_i)}.
	$$
	We can check analytically that the function $f : x \mapsto 2\sqrt{h_2^{-1}(x)}$ is concave on the interval $[0,0.4]$.
	Since, by assumption, $a_i$ lies in this interval for all $i \leq k$, we have  by concavity of $f$:
	\begin{align*}
		\sum_{i=1}^{k} b_i & \le \sum_{i=1}^{k} 2 \sqrt{h_2^{-1}(a_i)} = k \sum_{i=1}^{k} \frac{1}{k} 2\sqrt{h_2^{-1}(a_i)}  \le 2k \sqrt{h_2^{-1}\left(\frac{\CIC_{\calU_0}^0(\overline \Pi)}{k}\right)}.
	 \end{align*}
\end{proof}
We now conclude the proof of Proposition~\ref{prop:ICeps}.
If there exists $j \leq k$ such that $a_j > 0.4$, then we immediately have 
$$
\CIC_{\calU_0}^0(\overline \Pi) = \sum_{i=1}^k a_i \ge a_j > 0.4 \ge \frac{(1-2\eps)^2 \log(k)}{8k}.
$$
We now consider the case where $a_i \le 0.4$ for all $i \leq k$. We combine Claims~\ref{claim:deltakeps},~\ref{claim:deltakbi} and~\ref{claim:biQIL} to obtain
$$
\frac 12 (1 - 2\eps) \le 2k \sqrt{h_2^{-1}\left(\frac{\CIC_{\calU_0}^0(\overline \Pi)}{k}\right)},
$$
which implies 
	\begin{equation*}
		 \frac{\CIC_{\calU_0}^0(\overline \Pi)}{k} \ge h_2\left(\frac{(1-2\eps)^2}{16k^2}\right).
	\end{equation*}
	From there, we obtain, using the fact that $h_2(x) \geq x \log(1/x)$ for any $x\in[0,1]$~: 
	\begin{equation*}
		\frac{\CIC_{\calU_0}^0(\overline \Pi)}{k} \ge h_2\left(\frac{(1- 2\eps)^2}{16k^2}\right) \ge \frac{(1-2\eps)^2}{16k^2} \log\left(\frac{16k^2}{(1-2\eps)^2}\right) \ge \frac{(1-2\eps)^2 \log(k)}{8k^2}.
	\end{equation*}
	This yields $\CIC_{\calU_0}^0(\overline \Pi) \ge \frac{(1-2\eps)^2 \log(k)}{8k}$.	
	We conclude the proof of Proposition~\ref{prop:ICeps} using Lemma~\ref{mland:lem:CICQIL0}.
\end{proof}

\subsection{Tightness of the bound}
\label{qvariants-lbtight}
In~\cite{MR3473340-BravermanGKMT-2015}, the authors provide a protocol, denoted here $\Pi_{\AND}$  and described at Figure~\ref{fig:algorithmAND}, attributed to Jain, Radhakrishnan and Sen. It was proved to compute $\AND$ correctly under the distribution $\calU_0$ and to have information cost $O\big(\log(k)/k\big)$.

\begin{figure}[ht!]
\begin{center}
\fbox{
\begin{minipage}{13 cm}
\underline{\textbf{Protocol $\Pi_{\AND}$.}} Inputs : $(x,y) \in \zo\times\zo$. Parameter $r \in \NN$.
\begin{enumerate}
	\item Set $\theta = \frac{\pi}{8r}$. Let $\ket v$ be the vector $\cos(\theta) \ket 0 + \sin(\theta) \ket 1$. Let $U_v$ be the unitary operation of reflecting about the vector $\ket v$ i.e. $U_v \ket 0 = \cos(2\theta) \ket 0 + \sin(2\theta) \ket 1$ and $U_v \ket 1 = \sin(2\theta) \ket 0 - \cos(2\theta) \ket 1$. Also let $Z$ be the unitary operation of reflecting about $\ket 0$, i.e. $Z \ket 0 = \ket 0$ and $Z \ket 1 = - \ket 1$.
	\item Alice starts by preparing a qubit $C$ in state $\ket 0$.
	\item If $x = 0$, Alice applies the identity operation on $C$ and sends it to Bob. If $x = 1$, Alice applies the $U_v$ operation on $C$ and sends it to Bob.
	\item If $y = 0$, Bob applies the identity operation on $C$ and sends it to Alice. If $y = 1$, Bob applies the $Z$ operation on $C$ and sends it to Alice.
	\item After $k = 4r - 1$ rounds, Bob measures the register $C$. If the result is $1$, then he answers $1$, otherwise he answers $0$. He also sends this to Alice.
\end{enumerate}
\end{minipage}
}
\end{center}
\caption{Quantum protocol for $\AND$.}\label{fig:algorithmAND}
\end{figure}
Since $\Pi_{\AND}$ does not use any private memory register, their result and the lower bound provided by Proposition~\ref{prop:ICeps} yield Theorem~\ref{Theorem:CIC_pure}.

\section{Classical coins allow for private quantum protocols}
In this section, we study the behavior of $\CIC$ when the players are allowed to have private classical coins. 
We show that we can transform any quantum protocol into a memoryless quantum protocol with private coins whose $\CIC$ is exactly the information Bob learns just from the output of the protocol and nothing more.

Let us denote by ${\mathcal T}_{\mu,\epsilon,k}^{\memls,\coins}(f)$ the set of all $k$-round memoryless protocols $ \Pi$ with private coins computing function $f$ with error at most $\epsilon$ with respect to $\mu$.

We define the following notions for $\CIC$ with private coins.

\begin{definition}\label{def:QIC2}
	The $k$-round $\epsilon$-error private-coin memoryless classical input information complexity of $f$ relative to input distribution $\mu$ is defined as
	$
	\CIC_{\mu,\epsilon,k}^{\memls,\coins}(f) = \inf_{\Pi \in {\mathcal T}_{\mu,\epsilon,k}^{\memls,\coins}(f)} \CIC_{\mu}(\Pi).
	$
	We also define the quantity $\CIC_{\mu,\epsilon,k}^{\coins}(f)$ where we remove the memoryless constraint.
\end{definition}

Here, we will prove the following (restated) result.

\zeroCICPif*

We say that such a protocol $\Pi'$ is private since it leaks the minimal $\CIC$ possible: any protocol leaks at least as much information as Bob's output leaks to himself.

\begin{proof}

Let us consider a $k$-round quantum protocol $\Pi$ where Alice and Bob respectively have classical input $x$ and $y$. At the end of the protocol, Bob has an output in quantum register $\mathcal{O}_B = \mathcal{O}_k$, which he measures in the computational basis to get $f(x,y)$. Alice and Bob also share at the end of the protocol some garbage state in registers $\mathcal{G}_A, \mathcal{G}_B$. We have $\mathcal{G}_A = A_k$ and $\mathcal{G}_B = (B_k , \mathcal{O}_k) = (B_k ,  \mathcal{O}_B)$.

If at the end of $\Pi$ the output register $\mathcal{O}_B$ contains a superposition and not a mixture, we consider that Bob has an additional register $\mathcal{O}'_B$ and applies a CNOT operation on $\mathcal{O}_B,\mathcal{O}'_B$, so that $\mathcal{O}_B$ is a mixture. 
Bob now uses $\mathcal{O}_B$ as his output register. This protocol succeeds with the same probability as $\Pi$, so for notational simplicity we also denote it $\Pi$.

We now transform protocol $\Pi$ into a protocol $\overline \Pi$ which uses private classical coins and is a private protocol. In $\overline \Pi$, Alice and Bob will have some extra random coins. Let $m = \sum_{i = 1}^{k} |M_i|, a = \sum_{i = 1}^{k} |A_i|, b = \sum_{i = 1}^{k} |B_i|$ the total number of qubits sent in $\Pi$, and the cumulated memory sizes used by Alice and Bob respectively. Let $t = 2(m+a+b)$. In $\overline \Pi$, we give Alice and Bob respectively $t_A = (t_A^1,\dots,t_A^k)$ and $t_B = (t_B^1,\dots,t_B^k)$ random coins, where $t_A^1,t_B^1$ are of length $2(|A_0| + |M_0|)$ and $t_A^i,t_B^i$ are of length $2(|M_i|+|A_i|+|B_i|)$, $2 \leq i \leq r$. We also give Bob extra randomness $s_B$ of length $2 (|A_k| + |B_k| + |M_k|)$ and Alice extra randomness $s_A$ of length $2 (|A_k| + |B_k| + |M_k| - 1)$. The coins $t_A$ and $t_B$ will be used to perform an encrypted version of $\Pi$ while the coins $s_A,s_B$ will be used to retrieve the output with minimum information leakage.

Each player will encrypt the message he sends using the quantum one time pad. The quantum one time pad (QOTP) acting on $q$ qubits is the unitary $\mathcal{E}$ taking $2q$ classical bits $s = (s_1,\dots,s_{2q})$ as parameters such that 
$$ \mathcal{E}_{s}(\ket{\phi}) = \left(\bigotimes_{i = 1}^{q} X^{s_{2i - 1}} Z^{s_{2i}} \right) \ket{\phi}$$
where $X$ and $Z$ are respectively the bit flip operator and the phase flip operator. Note that the encrypted state when we average uniformly over $s$ is the totally mixed state.

We can now present our protocol $\overline \Pi$:
\begin{enumerate}
	\item \textbf{Alice and Bob perform an encrypted version of $\Pi$ using coins $t_A$ and $t_B$}: 
	\begin{itemize}
		\item At round $1$, Alice applies the same unitary $U_1$ as in $\Pi$ on registers $A_{0}  M_{0}$ to obtain $A_{1}  M_{1}$. Then, she encrypts register $A_{1}  M_{1}$ by applying the \qotp $\mathcal{E}_{t_A^1}$ on these registers. Finally, she sends $A_{1}  M_1$ to Bob.
		\item For $i \in \{2, \dots, k\}$, 
		\begin{itemize}
        		\item if $i$ is even, Bob applies $(\mathcal{E}_{t_B^{i-1}})^{-1}$ on register $A_{i-1} M_{i-1} B_{i-1}$ ($(\mathcal{E}_{t_B^{1}})^{-1}$ will act only on $A_1 M_1$) as an attempt to undo Alice's encryption. Bob then applies the same unitary $U_i$ as in $\Pi$ on registers $ M_{i-1} B_{i-1} $ to obtain $  M_{i} B_{i}$. We have here $A_{i} = A_{i-1}$. Then, he encrypts registers $A_{i} M_{i} B_{i}$ by applying the \qotp $\mathcal{E}_{t_B^i}$ on these registers. Finally, he sends $A_i M_{i} B_i$ to Alice.
        		\item if $i$ is odd, Alice applies $(\mathcal{E}_{t_A^{i-1}})^{-1}$ on register $A_{i-1} M_{i-1} B_{i-1}$ as an attempt to undo Bob's encryption.
        		Alice then applies the same unitary $U_i$ as in $\Pi$ on registers $A_{i-1} M_{i-1}$ to obtain  $A_{i} M_i B_{i}$. Here, we have $B_{i} = B_{i-1}$. Then, she encrypts registers $A_{i} M_i B_{i}$ by applying the \qotp $\mathcal{E}_{t_A^i}$ on these registers. Finally, she sends $A_{i} M_i B_i$ to Bob.
		\end{itemize}
		\item Finally, Bob applies $(\mathcal{E}_{t_B^{k}})^{-1}$ on registers $A_{k} M_{k} B_{k}$ as an attempt to undo Alice's encryption. Bob then applies $U_{k+1}$ on register $B_k = \mathcal{G}_B \mathcal{O}_B$ to get his output. 
	\end{itemize}
At this point, notice here that, if for all $i=1,\dots,k$, we have $t_A^i = t_B^i$, then Alice and Bob essentially performed $\Pi$.
\item \textbf{Bob sends his registers after encryption}:  Bob applies $\mathcal{E}_{s_B}$ on all the registers he holds, namely $A_k M_k B_k$, and sends them all to Alice. 
\item \textbf{Checking the coins $\mathbf{t_A = t_B}$}: Bob sends $t_B$ to Alice. Alice sends one bit to Bob: $1$ if $t_A = t_B$, $0$ otherwise. We distinguish now two cases:
\begin{itemize}
	\item if $t_A = t_B$, Alice applies $\mathcal{E}_{s_A}$ on registers $A_kM_k\mathcal{G}_B$, and sends back all her registers $A_kM_kB_k = A_kM_k\mathcal{G}_B\mathcal{O}_B$ to Bob. This means that all registers are encrypted by $s_A$ except $\mathcal{O}_B$ which is encrypted only by Bob. Bob undoes the encryption for this qubit using the corresponding bits in $s_B$ and outputs $\mathcal{O}_B$. 
	\item if $t_A \neq t_B$, Alice's state is totally mixed from Bob's randomness $s_B$ and Bob has no quantum registers. Then, both players discard their quantum registers and start again from step $1$, using fresh randomness for $t_A,t_B,s_B$. 
\end{itemize}
\end{enumerate}

\begin{remark}
	Note that we can also start from a quantum protocol $\Pi$, first transform it into a memoryless quantum protocol $\tilde \Pi$ as for instance in~\cite{BCG+16} and then define a private protocol $\tilde \Pi$ where the players encrypt and decrypt $\tilde \Pi$'s registers as described above. However, this would be at the expense of a quadratic blow-up in the communication cost while going from $\Pi$ to $\tilde \Pi$.
\end{remark}

\noindent \textbf{Analysis of the protocol}

\begin{enumerate}
\item \textbf{Encrypted $\Pi$}:
For an odd round $i$, Alice sends registers $A_i M_i B_i$ to Bob. We denote by $T_A,T_B,S_B$ the registers corresponding to $t_A,t_B,s_B$ respectively. The term appearing in CIC for this round is: 
\begin{align}\label{Eqn:1}
	I(A_i M_i B_i : X | T_B S_B Y)_{\rho^{A_i M_i B_i X Y T_B S_B}}
	\leq
	I(A_i M_i B_i : X Y T_B S_B)_{\rho^{A_i M_i B_i X Y T_B S_B}}
\end{align}
where $\rho^{A_i M_i B_i X Y T_B S_B}$ is of the form $\mathbb{I}_{A_i M_i B_i} \otimes \rho^{X Y T_B S_B}$ (for some classical state $\rho^{X Y T_B S_B}$),
because of the quantum one time pad realized by Alice using her coins, $T_A$. Hence $I(A_i M_i B_i:X|T_B S_B Y) = 0$. We can prove a similar statement for even $i$.
\item \textbf{Encrypting and sending Bob's registers}:
As above, this consists only of sending encryptions back and forth of quantum registers using fresh randomness. Similarly as in Equation \ref{Eqn:1}, this does not give any information.
\item \textbf{Checking the coins}: 
When Bob sends $t_B$ to Alice, her quantum state is totally encrypted by coins $s_B$ (which are kept by Bob) hence her state is independent of $t_B$. Therefore, she does not receive any information about $Y$ when she learns $t_B$. 
Then Alice checks whether $t_A = t_B$ or not, and she sends the answer to Bob. 
\begin{itemize}
	\item If $t_A \neq t_B$, then Bob knows only this information, which is independent of $x$ so the term in CIC is zero. Then the players restart the process.
	\item If $t_A = t_B$, Alice's quantum registers are fully encrypted by Bob which, as usual, does not contribute to $\CIC$. After Alice's last message to Bob, let $\rho_{final}$ the shared state. The contribution to $\CIC$ is 
	
	\begin{align*}
	 I(A_kM_k\mathcal{G}_B\mathcal{O}_B : X |T_BS_BY)_{\rho_{final}} & = 
	 I(\mathcal{O}_B : X |T_BS_BY)_{\rho_{final}} + I(A_kM_k\mathcal{G}_B : X |T_BS_BY\mathcal{O}_B)_{\rho_{final}} \\
	 & \le I(\mathcal{O}_B : X |T_BS_BY)_{\rho_{final}} + I(A_kM_k\mathcal{G}_B : X T_BS_BY\mathcal{O}_B)_{\rho_{final}} \\
	 & = I(\mathcal{O}_B : X |T_BS_BY)_{\rho_{final}}
	 \end{align*}
	where we used 
	$\rho_{final} = \rho_{final}^{A_kM_k\mathcal{G}_B} \otimes \rho_{final}^{X T_BS_BY\mathcal{O}_B} = 
	\mathbb{I}_{A_kM_k\mathcal{G}_B}  \otimes \rho_{final}^{X T_BS_BY\mathcal{O}_B} $ 
	because of Alice's final encryption.
	We are in the case $t_A = t_B$ so the output register is exactly the one of the original protocol, and therefore contributes $I\left(\Pi_{out}(X,Y) :X|Y\right)$ to the total $\CIC$.
\end{itemize}
The above analysis can be carried out at each repetition of these steps until the players have a run where they agree on their coins. Hence $\overline \Pi$ has minimal $\CIC$.
\end{enumerate}
\end{proof}
 
Our results imply that given any function $f$, if we take for $\Pi$ the protocol where Alice just sends over $x$ to Bob who computes and output $f(x,y)$, we obtain a protocol $\Pi'$ that can perfectly compute $f$ with $\CIC$ only the information gained from $f(x,y)$. Hence, we have the following Corollary.

\zeroCICPifcorofct*

Note that, in the case of AND, on distribution $\calU_0$ the output of $\AND(x,y) = x \wedge y$ is always $0$. Hence, by the above result, $\CIC^{\memls,\coins}_{\calU_0,k}(\AND) = \CIC^{\coins}_{\calU_0,k}(\AND) = 0$ for every integer $k$.


\section{One-shot coins}

\subsection{Definitions}\label{Section:Memoryless_Definitions}

In the previous section we saw that the model where quantum players have private coins allows for private protocols. While this implies that quantum communication can offer a new advantage compared to classical communication, it also shows that information complexity cannot provide non-trivial lower bounds for the quantum communication in this model. For this reason, we define a new model, where we restrict how the players can use these private coins. More precisely, we allow the players to use one-shot private coins, meaning that at each round they can only use fresh private coins which are independent of the previous and the future rounds. 

\begin{definition}
Consider a quantum protocol with private coins. A coin is said to be one-shot if it is read only once or, in other words, if it used only once by a unitary operation during any run of the protocol. The protocol is said to be one-shot coin if all its coins are one-shot.  
\end{definition}

One-shot coins are a very natural way of using coins. In the classical setting, every randomized protocol (without restriction on coins) can be simulated by a one-shot coin protocol with the same information and communication costs, since it is always possible to sample directly the messages~\cite{LauriereT16}.

More precisely, a memoryless quantum protocol $\Pi$ with one-shot coins can be described as follows. At each round, a fresh set of coins is used. For odd (resp. even) $i \in [k]$, we denote by $R^A_i$ (resp $R^B_i$) the random variables corresponding to the coins used by Alice (resp. Bob) at round $i$. We note $R^A = (R^A_i)_{i \in [k], i \,\mathrm{ odd}}$ and $R^B = (R^B_i)_{i \in [k], i \,\mathrm{ even}}$. For simplicity we use the same notation for the (classical) registers corresponding to these random variables. The interaction will be the following:
\begin{itemize}
	\item The players start with classical respective registers $X,R^A$ and $Y,R^B$. Alice starts with a quantum register $M_0$.
	\item At any odd round $i$, Alice applies an isometry from $M_{i-1}$ to $M_i$ using $R^A_i$ as a classical control, and sends the whole quantum register $M_i$ to Bob.
	\item At any even round $i$, Bob applies an isometry from $M_{i-1}$ to $M_i$ using $R^B_i$ as a classical control, and sends the whole quantum register $M_i$ to Alice.
\end{itemize}

Let us denote by ${\mathcal T}_{\mu,\epsilon,k}^{\memls,\oscoins}(f)$ the set of all $k$-round memoryless protocols $ \Pi$ with private one-shot coins computing function $f$ with error at most $\epsilon$ with respect to $\mu$.

\begin{definition}\label{def:QIC_C1}
	The $k$-round $\epsilon$-error one-shot-coin memoryless classical input information complexity of $f$ relative to distribution $\mu$ is defined as
	$
	\CIC_{\mu,\epsilon,k}^{\memls,\oscoins}(f) = \inf_{\Pi \in {\mathcal T}_{\mu,\epsilon,k}^{\memls,\oscoins}(f)} \CIC_{\mu}(\Pi).
	$
	We also define the quantity $\CIC_{\mu,\epsilon,k}^{\oscoins}(f)$ where we remove the memoryless constraint.
\end{definition}


\subsection{Relating $CIC^{\memls}_{\calU_0,k}$  and $\CIC^{\memls, \oscoins}_{\calU_0,k}$ for the $\AND$ function}


In this section, we show how to remove the one-shot coins from a memoryless quantum protocol with input distribution $\calU_0$ without increasing too much the information cost of the protocol.

\oneshotcoinsBound*

Before proving this result (see Section~\ref{mland:sec:oneshot-proof} below), let us stress that
for protocols computing (exactly) AND under distribution $\calU_0$, this implies immediately 
$$
CIC^{\memls}_{\calU_0,\eps,k}(AND) 
	=
	O\left(\CIC^{\memls, \oscoins}_{\calU_0,\eps,k}(AND) \cdot \left(\log(k) + \left| \log \CIC^{\memls, \oscoins}_{\calU_0,\eps,k}(AND) \right| \right) \right).
$$

Combining the above with Theorem~\ref{Theorem:CIC_pure}, we obtain 
\begin{corollary}
	For every positive integer $k$ and $\eps>0$, $ \CIC^{\memls, \oscoins}_{\calU_0,\eps,k}(AND) = \Theta_\eps\left(\frac{1}{k}\right).$
\end{corollary}

This also implies that one-shot coins are very different than private coins in the quantum setting, since for the AND function, private coins allow for a private protocol, while one-shot coins do not.

\subsection{Proof of Theorem~\ref{mland:thm:oneshotres}}\label{mland:sec:oneshot-proof}

As we said, we will transform a memoryless quantum protocol $\Pi$ with one-shot coins into a memoryless quantum protocol $\Pi'$ without any coins, such that $\Pi'$ will have the same outputs than $\Pi$ and without increasing the information cost too much. The transformation from $\Pi$ to $\Pi'$ will informally be the following:
\begin{enumerate}
	\item Quantize the coins from $\Pi$ \ie put them in quantum superposition in quantum registers $\widetilde{R}^A$ and $\widetilde{R}^B$.
	\item At each odd (or even) round, Alice (or Bob) applies the same transformation as in $\Pi$. Then, Alice (or Bob) would like to send all their quantum registers, including the coin registers, to the other player. Before doing that, Alice (or Bob) applies a compensation unitary that will limit the information cost increase that occurs because of the sending of all the quantum registers. 
\end{enumerate}

Step 1 is pretty straightforward. Step 2 will require a way to perform this compensation unitary, which will result in an increase in the information proportional to the binary entropy $h_2$ at each round. While $h_2$ is not convex, we show an weak-convexity result which will allow us to limit the total increase for all the rounds. The next subsection will deal with those two preliminary results.

\subsubsection{Two useful lemmata}
\begin{lemma} \label{Lemma:PurificationLemma}
	Consider a state 
	$ \rho = \frac{1}{2} \sum_{x \in \zo} \akb{x}_{\space{X}} \otimes \akb{\phi_x}_{\space{AB}} $
	for some pure states $\ket{\phi_x}$, $x \in \zo$. There exist two unitary operators $U_0,U_1$ acting on $\space{A}$ such that if we define 
	$$ \rho' = \frac{1}{2} \sum_{x \in \zo} \akb{x}_{\space{X}} \otimes (U_x \otimes I_B) \akb{\phi_x}_{\space{AB}} (U_x^\dagger \otimes I_B), $$
	we have $I(X:AB)_{\rho'} \le h_2\left( \frac{I(X:B)_{\rho}}{2} \right)$. 
\end{lemma}
\begin{proof}
	For each $x$,  let $\rho_x = Tr_{\space{A}} \akb{\phi_x}$. By Uhlmann's theorem, there exists a unitary $V$ acting on $\space{A}$ such that $\bra{\phi_0}  V  \ket{\phi_1} = F(\rho_0,\rho_1)$. Let $U_0 = I$ and $U_1 = V$. We have
	$$
	I(X:AB)_{\rho'} = h_2\left(\frac{1 - |\bra{\phi_0}  V  \ket{\phi_1}|}{2}\right) = h_2\left(\frac{1 - F(\rho_0,\rho_1)}{2}\right) \le h_2\left(\frac{I(X:B)_{\rho}}{2}\right),
	$$
where the first equality is by  Lemma~\ref{Lemma:PurePinsker}, the second equality is by definition of $V$, and the inequality is by Lemma~\ref{lem:mi-fid}.
\end{proof}

\begin{lemma} \label{Lemma:BinaryEntropyExtendedConvenxity}
	Let $x_1,\dots,x_n \in [0,1]$ and let $S = \sum_{i} x_i$.
	We have
	$$ \sum_{i = 1}^n h_2(x_i) \le O(S\log(n) + S | \log(S)|)$$
	with the convention that $0 \log(0) = 0$.
\end{lemma}
\begin{proof}
	If $S=0$, then $x_i = 0$ for all $i$ hence the property is true. From now on, assume $S \neq 0$.
	Let $U = \frac{2n}{S}$ and $x \in [0,1]$. Since $U  \ge 2$, we have 
	\begin{itemize}
		\item if $x \ge \frac{1}{2}$, then $h_2(x) \le 1 \le 2x$
		\item if $x \in [\frac{1}{U},\frac{1}{2}]$, then $h_2(x) \le 2x|\log(x)| \le 2x|\log(U)|$
		\item if $x \in [0,\frac{1}{U}]$, then $h_2(x) \le h_2(\frac{1}{U}) \le 2\frac{|\log(U)|}{U}$.
	\end{itemize}
Since $U \ge 2$, for all $x \in [0,1]$, we have $2x \le 2x|\log(U)|$ hence 
\begin{align}
h_2(x) \le \max \left\{2x|\log(U)|,2\frac{|\log(U)|}{U} \right\} \le 2|\log(U)|\left(x + \frac{1}{U}\right).
\end{align}
Using the above for each $x_i$ and summing over $i$, we obtain 
\begin{align*}
\sum_{i = 1}^n h_2(x_i) \le 2|\log(U)| \left(S + \frac{n}{U}\right).
\end{align*}
Plugging the value of $U$ in the above expression yields 
$$ \sum_{i = 1}^n h_2(x_i) \le 3 S \left| \log\left( \frac{2n}{S} \right) \right| \le O(S\log(n) + S | \log(S)|).$$
\end{proof}

\subsubsection{Constructing $\Pi'$ from $\Pi$} 
We now prove Theorem~\ref{mland:thm:oneshotres}. 

\begin{proof}[Proof of Theorem~\ref{mland:thm:oneshotres}]
Consider a memoryless quantum protocol $\Pi$, described as in Section \ref{Section:Memoryless_Definitions}.  Recall that:
$$ \CIC_{\calU_0}(\Pi) = \sum_{i \ odd} I(X:M_i|Y=0,R^B) + \sum_{i \ even} I(Y:M_i|X=0,R^A).$$

Because the coins are independent of the inputs, we have 

\begin{align}\label{Eq:OriginalCIC}
\CIC_{\calU_0}(\Pi) & = \sum_{i \ odd} I(X:M_iR^B|Y=0) + \sum_{i \ even} I(Y:M_iR^A|X=0) \\
& \ge \sum_{i \ odd} I(X:M_i|Y=0) + \sum_{i \ even} I(Y:M_i|X=0).
\end{align}

\paragraph{Construction of $\Pi'$.} 
We start by quantizing the coins of $\Pi$ using two quantum registers $\widetilde{R}^A$ and $\widetilde{R}^B$. 
\begin{itemize}
	\item At the beginning of the protocol, Alice starts with registers $\widetilde{R}^A \otimes M_0 \otimes \widetilde{R}^B$.
	\item At each odd round $i$, Alice performs, as in $\Pi$, the isometry $U_i$ on $\widetilde{R}^A_i   M_{i-1}$ to obtain $\widetilde{R}^A   M_i$. This isometry uses the coin register only as a control string.  Let $\rho_i$ be the state of registers $XY\widetilde{R}^A_{\le i}   \widetilde{R}^B_{\le i}   M_i    \widetilde{R}^A_{> i} \widetilde{R}^B_{> i}$. Alice holds all of them except $Y$. The registers $\widetilde{R}^A_{> i}   \widetilde{R}^B_{> i}$ have not been used yet by any player and contain a pure state independent of the players' other registers. Therefore, $I(X:\widetilde{R}^A_{> i}   \widetilde{R}^B_{> i} | Y=0) = 0$ and $I(X:M_i   \widetilde{R}^A_{> i}   \widetilde{R}^B_{> i} | Y=0)_{\rho_i} =  I(X:M_i | Y=0)_{\rho_i}$. Hence, we can apply Lemma~\ref{Lemma:PurificationLemma}, and obtain the existence of a unitary operation $V^x_i$ on $\widetilde{R}^A_{\le i}   \widetilde{R}^B_{\le i}$ such that the resulting state \\ $\rho'_i = (I_{XY} \otimes V^x_i \otimes I_{M_i    \widetilde{R}^A_{> i} \widetilde{R}^B_{> i}}) \rho_i (I_{XY} \otimes V^x_i \otimes I_{M_i    \widetilde{R}^A_{> i} \widetilde{R}^B_{> i}})^\dagger$ satisfies
	$$
	I(X:\widetilde{R}^A_{\le i}   \widetilde{R}^B_{\le i}   M_i    \widetilde{R}^A_{> i}   \widetilde{R}^B_{> i} | Y=0)_{\rho'_i} \le h_2\left( \frac{I(X:M_i \widetilde{R}^A_{> i} \widetilde{R}^B_{> i} | Y=0)_{\rho_i}}{2} \right) = h_2\left( \frac{I(X:M_i | Y=0)_{\rho_i}}{2} \right). 
	$$ 
	\item At each even round $i$, we do exactly the same thing from Bob's side. The resulting state $\rho'_i$ of $\Pi'$ after step $i$ satisfies 
	$$
	I(Y:\widetilde{R}^A_{\le i}   \widetilde{R}^B_{\le i}   M_i    \widetilde{R}^A_{> i}   \widetilde{R}^B_{> i} | X=0)_{\rho'_i} \le h_2\left( \frac{I(Y:M_i \widetilde{R}^A_{> i} \widetilde{R}^B_{> i} | X=0)_{\rho_i}}{2} \right) = h_2\left( \frac{I(Y:M_i | X=0)_{\rho_i}}{2} \right). 
	$$
	\item Protocol $\Pi'$ acts on the message registers similarly as in $\Pi$ and Bob outputs therefore the same as in $\Pi$.
\end{itemize}

Let us now calculate the $\CIC$ of our new protocol $\Pi'$. Let $k$ the number of rounds of $\Pi$ (and hence $\Pi'$). 

\begin{align*}
\CIC_{\calU_0}(\Pi') & = \sum_{i \ odd} I(X:\widetilde{R}^A_{\le i}   \widetilde{R}^B_{\le i}   M_i    \widetilde{R}^A_{> i}   \widetilde{R}^B_{> i} | Y=0)_{\rho'_i} + \sum_{i \ even} I(Y:\widetilde{R}^A_{\le i}   \widetilde{R}^B_{\le i}   M_i    \widetilde{R}^A_{> i}   \widetilde{R}^B_{> i} | X=0)_{\rho'_i} \\
& \le \sum_{i \ odd} h_2\left( \frac{I(X:M_i | Y=0)_{\rho_i}}{2} \right) + \sum_{i \ even} h_2\left( \frac{I(Y:M_i | X=0)_{\rho_i}}{2} \right) \\
\end{align*}

Let $x_i = I(X:M_i | Y=0)_{\rho_i}$ for $i$ odd and $x_i = I(Y:M_i | X=0)_{\rho_i}$ for $i$ even. Notice that $\rho_i$ coincides with the state in $\Pi$ on registers $X,M_i,Y$ so those terms are exactly those of Equation \ref{Eq:OriginalCIC}, which gives $S := \sum_i x_i \le \CIC_{\calU_0}(\Pi).$ Using Lemma \ref{Lemma:BinaryEntropyExtendedConvenxity}, we can conclude that 
$$ \CIC_{\calU_0}(\Pi') \le O(S\log(k) + S | \log(S)|) \le O\Big(\CIC_{\calU_0}(\Pi) \cdot \Big( \log(k) +  \left| \log \CIC_{\calU_0}(\Pi)\right| \Big) \Big).$$
\end{proof}

\paragraph{Acknowledgements} 
The authors would like to thank Rahul Jain, Virginie Leray, Ashwin Nayak and Dave Touchette for fruitful discussions. I. K. was supported by the ERC project QCC. 

\bibliography{memoryless-qcc}


\end{document}